\documentclass[11pt]{article}

\usepackage[utf8]{inputenc}
\usepackage[T1]{fontenc}
\usepackage[english]{babel}

\usepackage{amsmath, amssymb, amsthm}
\usepackage{mathrsfs}
\usepackage{stmaryrd}
\usepackage{verbatim}
\usepackage{enumerate}

\usepackage{geometry}
\usepackage{hyperref}
\hypersetup{
  colorlinks=true,
  linkcolor=blue,
  citecolor=blue,
  urlcolor=black
}

\usepackage{graphicx}
\usepackage{enumitem}
\usepackage{booktabs}
\usepackage{xcolor}

\usepackage{cite}

\usepackage{authblk}
\newtheorem{theorem}{Theorem}[section]

\theoremstyle{remark}

\theoremstyle{definition}
\newtheorem{definition}[theorem]{Definition}

\usepackage{fancyvrb}
\VerbatimFootnotes
\fvset{listparameters={\setlength{\topsep}{1.5pt}\setlength{\partopsep}{0pt}}}

\usepackage{float}
\usepackage{caption}

\newfloat{listing}{htbp}{lol}
\floatname{listing}{Listing}
\title{Provenance-Enhanced Statements in Knowledge Graphs}

\author{Fabio Vitali\thanks{fabio.vitali@unibo.it}}
\affil{University of Bologna, Italy}
\author{Valentina Pasqual}
\affil{University of Bologna, Italy}
\date{June 2026}

\begin{document}
\maketitle

\begin{abstract}
Provenance-enhanced statements of the form "according to $X$, $\varphi$" are pervasive in contemporary knowledge graphs, especially in domains where graph content primarily represents claims, interpretations, and hypotheses (\emph{capta}) rather than observer-independent facts (\emph{data}). Current provenance models can record who asserted what, but they typically treat provenance as semantically neutral, leaving underspecified how attributed claims relate to factual commitment, to one another, and to reasoning.

In this paper we introduce DEC, a framework that interprets provenance predicates as indicators of epistemic stance and groups provenance-homogeneous sets of statements into \emph{cognitive worlds}. Drawing on cognitive modal logics (doxastic, epistemic, and conjectural), DEC characterizes locality, rationality, and controlled permeation between cognitive worlds and a distinguished factual core (“reality”), thereby enabling principled reasoning over attributed content without collapsing disagreements into inconsistencies.

We formalize a DEC interpretation for RDF datasets that is conservative over RDF~1.2 semantics, clarify the role of intensionality and identity (including the Superman paradox), and illustrate the approach on common Semantic Web representations (named graphs, quoted triples/RDF-star, and reification)
. Finally, we describe our prototype DEC reasoner implemented as a Fuseki dataset module, supporting controlled factualisation and explicit detection of disagreements and delusions.
\end{abstract}

\begin{flushleft}
\end{flushleft}

\section{Introduction}\label{sec:introduction}
In this paper we address the problem of provenance-enhanced statements (PES) and of their semantic characterization beyond what is currently achievable in mainstream Semantic Web technologies. Specifically, we aim to study sentences such as "According to Martin Kemp, the painting \textit{Salvator Mundi} was painted by Leonardo da Vinci" and discuss characteristics and differences in the factuality of sentences (their \textit{epistemic status}) by expanding the neutral "According to Martin Kemp..." into different types of provenance predicates (e.g., "Martin Kemp wrote...", "Martin Kemp believed...", "Martin Kemp discovered..." or "Martin Kemp supposed..."). 

Provenance-enhanced statements constitute a substantial portion of the content of large and widely used knowledge graphs. They capture information for which provenance does not merely provide contextual metadata, but plays a decisive role in justifying the very presence of a statement in the graph. This distinction is closely related to a classical opposition, that between \textit{data} and \textit{capta} \cite{Drucker11}. Data denote states of affairs that are taken to hold independently of any particular observer or interpretation. \textit{Data} are objective, stable and independent of personal interpretation: for instance, the physical dimensions of a painting or its material. \textit{Capta} on the contrary denote what is apprehended, recorded, or asserted by cognitive agents expressing their understanding about something: for instance, the authorship of a painting, its creation date, its style, etc. 

 For data, the identification of provenance is largely irrelevant: the objective nature of the fact makes its source essentially informative. For capta, instead, provenance is culturally significant: the identification of who made a statement, when, and under which authority is essential to assess its relevance, reliability, and acceptance.
 
 Current knowledge graphs overwhelmingly deal with capta; clearly, assertions such as the attributions of authorship, the datings, the classifications of artefacts and cultural heritage items, are not objective measures or evidence, but the outcomes of scholarly analysis, which are typically reported together with their sources and justifications. In these contexts, disagreement is not an anomaly but a structural feature of the domain. The epistemic state of capta is also not uniform: some may be widely accepted, others contested, provisional, hypothetical, abandoned, illusory, or idiosyncratic. In general, capta are not “true” in an absolute sense: at best, they can be seen as broadly accepted within a community. 
From this perspective, knowledge graphs representing only a single, allegedly factual view are losing relevant information: the divergence carries meaning, reflects the state of a debate, the multiplicity of opinions, the evolution of scholarly consensus. A model that aims to support analysis, comparison, and interpretation must therefore represent such divergence explicitly. This produces a complex representational landscape that cannot be adequately captured by treating provenance as a purely ancillary annotation.

Current KG approaches provide only limited support for this task. Models such as PROV-O treat provenance as an external addition to an otherwise unqualified statement, without addressing its epistemic status. Systems such as Wikidata introduce rankings and uncertainty indicators, but do so in a coarse and largely informal manner, without a clear semantic characterisation of the relationship between statements' provenance and reality.

\bigskip

In this paper we pursue three main objectives.
First, we aim to characterize epistemically the different relations that may hold between facts and statements qualified by provenance.
Second, we propose a small but plausible set of epistemic characterizations, and we associate them with \textit{cognitive worlds}, i.e., homogeneous sets of statements sharing similar provenance.
Third, we identify and analyse a number of epistemically relevant phenomena that emerge from this framework, including disagreements, delusions, settlements, and subjective (or personal) opinions.

To achieve this, we adopt cognitive modal logics as our conceptual and formal reference, in particular \textit{doxastic, epistemic},  and \textit{conjectural logics} (hereafter referred to as \textit{DEC}). These logics provide a foundation for modelling epistemic statuses and the relationship between statements and the agents or sources that support them.

Cognitive modal logics supply a number of concepts and mathematical tools that are especially suited to the characterization of cognitive worlds. Among these we list \textit{locality, rationality, intensionality}, and particularly \textit{permeation}. By selectively enabling or constraining them, we are able to obtain a range of cognitive worlds with different relationships to one another and to facts. Additionally, this model supports reasoning over provenance-enhanced statements without the known difficulties that have discouraged reasoning over quoted or reified statements in the Semantic Web, e.g., the \textit{Superman paradox}. 
Beyond static characterization, the framework also supports a type of epistemic evolution of cognitive worlds. This allows the DEC framework to treat acceptance, rejection, and obsolescence of statements not as external events, but as \textit{settlements}, i.e., sequences of evolving epistemic states, transformations of cognitive worlds from  less determined to more determined ones.

A prototype DEC reasoner, implemented on top of Semantic Web technologies, has been created to prove and complement the theoretical DEC framework. The DEC Reasoner shows how cognitive worlds, permeation principles, and epistemic constraints can be put into practice in order to support reasoning over provenance-enhanced statements in a controlled fashion. It also shows that epistemic characterization of cognitive worlds is independent of any specific representational mechanism, supporting named graphs, quoted triples, and classical reification in a similar fashion. 

The paper is organized as follows. Section~\ref{sec:background} provides background on provenance-enhanced statements based on the data/capta distinction, and the limitations of semantically neutral provenance in mainstream KG practice. Section~\ref{sec:modallogics} introduces our reference family of cognitive modal logics (doxastic, epistemic, and conjectural), and details the key notions of locality, rationality, permeation and intensionality. Section~\ref{sec:DECframework} then defines the DEC framework for RDF datasets: it introduces cognitive worlds as provenance-homogeneous epistemic containers, specifies their interaction with a distinguished reality graph, and discusses how epistemically relevant phenomena such as settlement, disagreement, delusion, and personal opinions can be expressed and detected. Section~\ref{sec:decreasoner} describes a prototype DEC reasoner implemented on top of Apache Jena/Fuseki, showing how DEC can be operationalised as a supervisory layer over standard RDF/RDFS/OWL inference. Section~\ref{sec:relatedworks} discusses related work on contextual reasoning, statement-level metadata models, and epistemic approaches to provenance. Section~\ref{sec:conclusions} concludes and outlines directions for future work. 

\section{Provenance and Epistemic Characterization in KGs}\label{sec:background}

In knowledge graphs, provenance is commonly understood as the set of annotations that describe the origin of statements: who asserted them, from which source they were derived, when and how they were produced, and under which conditions they hold. Provenance plays a central role in large-scale, heterogeneous, and collaborative KGs, where data are aggregated from multiple sources with varying degrees of reliability.

The primary functions of provenance are well established: it supports assessments of trust, reliability, and accountability by linking claims to sources and agents; it enables versioning and change tracking, making it possible to record corrections, retractions, and refinements over time; it facilitates data integration, allowing systems to manage overlap and inconsistency across heterogeneous datasets. These roles are extensively documented in the literature and supported by mature standards and infrastructures (see, e.g., \cite{prov-o}, \cite{sikos2020kg}).

A range of technical mechanisms has been developed to attach provenance to RDF statements. Traditional approaches include RDF reification and named graphs; more recent ones include RDF-star \cite{rdfStarEditorsDraft}, now incorporated into RDF 1.2 \cite{CLK26}, and nanopublications, which explicitly separate assertion, provenance, and publication information into distinct graphs \cite{kuhn2013nanopubs}. PROV-O remains the most widely adopted conceptual model for representing provenance across these approaches \cite{prov-o}, while Wikidata relies on a statement-centric model in which each claim is explicitly associated with provenance (and possibly qualifiers and preference ranks) \cite{vrandecic2014}. In CIDOC-CRM, conversely, provenance is an integral part of the ontology, modeled through events and attribution acts performed by agents in time \cite{cidocCRM}.

Provenance is of extreme importance in the infrastructural well-being of KG frameworks. Systematic provenance management supports transparency, reproducibility, debugging, and incremental updates in complex KG pipelines \cite{kleinsteuber2024provenancePlatforms}. Industrial and highly regulated contexts rely on detailed provenance and traceability because knowledge graphs often use the current snapshot of data, and without comprehensive provenance and change history, traceability of changes, and recovery of past states, regulatory compliance cannot be guaranteed \cite{dibowski2024traceability}. Documented and verifiable provenance is fundamental for the quality assurance of knowledge graphs and their content, so much so that manual verification of large datasets is impractical and automated verification of provenance is necessary to ensure that the documented sources actually support the asserted graph facts \cite{amaral2023prove}. In addition to provenance, the effectiveness of different representation techniques, including named graphs and embedded triples, needs to be evaluated (e.g., see \cite{massari2025provenanceDH}) with particular attention to RDF conformance and scalability. 

However, these approaches deliberately adopt a neutral stance with respect to the semantic impact of provenance. As explicitly noted by Sikos and Philp \cite{sikos2020kg}, provenance annotations are intended to provide contextual information “without altering the semantic interpretation” of the underlying statements. In other words, attaching provenance to a triple is not supposed to affect its truth conditions.

While semantic neutrality simplifies modeling, it proves inadequate in many practical situations. A generic provenance-enhanced statement (PES) has the form “According to X, $\varphi$.” The truth of the attribution (“according to X”) and the truth of $\varphi$ are logically independent. A statement may be correctly attributed and true\footnote{These are all examples of \textit{capta} so obviously "true"/"false" should be read as "considered true"/"considered false".} (e.g., “\textit{According to Giorgio Vasari, Leonardo da Vinci painted Mona Lisa}”), incorrectly attributed but true (e.g., “\textit{According to Giorgio Vasari, Pablo Picasso painted Guernica}”\footnote{Giorgio Vasari (1511-1574), the best known biographer of Italian Renaissance painters, knew of Leonardo da Vinci (1452-1519), but died three centuries before the birth of Pablo Picasso (1881-1973)}), correctly attributed but false (“\textit{According to Archbishop Ussher, the Earth was created in 4004 BC}”), or both incorrectly attributed and false (e.g., “\textit{According to Albert Einstein, the Earth was created in 4004 BC}”). Existing provenance models typically avoid this complexity by leaving the truth status of the reified statement unspecified, and, when necessary, by asserting a non-reified version of the same content as factual.

For instance, in Wikidata all claims are represented as reified statements enriched with provenance, qualifiers, and ranks (preferred, normal, deprecated), while a single statement can be promoted as factual by asserting it again as an unreified triple \cite{vrandecic2014}. This approach preserves alternative claims without collapsing them, but at the cost of relegating epistemic distinctions to an informal ranking mechanism and to application-level interpretation \cite{dipasquale2024weakerLogicalStatusWikidata}.

\bigskip
From an epistemological perspective, the very act of adding provenance often signals limited epistemic commitment: if a claim were regarded as unproblematic and directly grounded in reality, there would be little reason to qualify it explicitly by stating who asserted it and under which conditions. Provenance thus functions not only as a technical device, but also as implicit evidence of uncertainty, contestability, and interpretation.

Natural language routinely exploits this dimension. A neutral PES "according to X, $\varphi$" can instead be found expressed as “X writes $\varphi$,” “X believes $\varphi$,” “X proves $\varphi$,” and “X supposes $\varphi$”, each of which conveys a  different relationship between $\varphi$ and factuality, and different expectations about the stability, revisability, and authority of $\varphi$. Neutral provenance relations flatten these distinctions, removing information that is crucial for interpretation and reasoning.

Epistemic characterization of PESs can be viewed through the distinction between \textit{data} and \textit{capta} \cite{Drucker11}. Data denote states of affairs taken to hold independently of interpretation: objective, stable, and observer-independent. Capta, by contrast, denote what is apprehended, selected, interpreted, and reported by cognitive agents. Much of what is commonly treated as data is really the result of interpretive choices; Gitelman’s well-known claim that “raw data is an oxymoron” captures the idea that data are inherently constructed artifacts, shaped by instruments, protocols, and classificatory decisions rather than unmediated representations of reality \cite{gitelman2013}; Kitchin generalizes this view, emphasizing the constructed nature of data across domains \cite{Kitchin_2014}: these opinions suggest (or, indeed, strongly claim) that capta largely outweigh data in most collections of data, scholarly or otherwise.   
Indeed, knowledge graphs, as they are constructed and curated in practice, overwhelmingly encode claims as they appear in scholarly  texts, archival catalogues, scientific publications, or databases, together with their sources and justificatory contexts, and disagreement is frequent, persistent and often irreducible, regardless of domains, humanities or hard sciences. 

Consider for instance the painting known as the \emph{Salvator Mundi}, which is to this day the subject of a prolonged debate regarding its authorship. The well-publicized attribution to Leonardo da Vinci coexists with no less than six other theories where expert positions favor other artists, such as Boltraffio or Salaì, two of the main Leonardo's pupils (e.g., see \cite{zollner2019salvatorMundi}, \cite{artnewspaper2017salvatorMundi}). Each attribution is supported by stylistic analysis, historical documentation, and scholarly discussion, and most probably no argumentation will conclusively establish a single attribution that can be universally accepted. 

Similar situations exist also in scientific domains, where factual objectivity is often assumed even without solid justification. We are not even referring to cases in which no relevant experiment has been performed yet, so a definite answer is obtainable but not obtained at present. On the contrary, we refer to situations in which the literature contains serious and methodologically credible lines of evidence pointing in opposite directions. For instance, for much of the XX century peptic ulcer disease was explained primarily in terms of stress, lifestyle, and excess gastric acid. In the early 1980s, a competing explanatory framework emerged, identifying \emph{Helicobacter pylori} as the causal factor \cite{marshall1984} and later still extensive clinical validation settled the debate and created a new consensus \cite{nih1994}. Yet, more recently, a third multifactorial framework has been proposed (but not yet universally accepted), in which \textit{H. pylori} remains the primary cause, but additional factors such as non-steroidal anti-inflammatory drugs (NSAIDs) and, again, stress are considered as contributing or aggravating elements \cite{malfertheiner2017}. For a significant period, these explanatory models coexisted. In many ways these theories are partially overlapping, partially incompatible, and supported by different bodies of evidence. Yet one was first accepted and then rejected by the scientific community and replaced by another only when new evidence was considered.   

Treating such claims as if they were data misrepresents their epistemic status and obscures the dynamics of knowledge production. A KG that collapses capta into a single factual view loses precisely the information that is most relevant for analysis and interpretation.

\bigskip

A further limitation of current provenance models lies in their focus on individual statements. Scholarly and scientific knowledge typically evolves through \textit{theories}, coherent sets of interdependent claims that are meant to be considered, accepted, revised, or abandoned only as a whole. For instance, Leonardo appears as an author of the Salvator Mundi in at least three different authorship theories (alone or in collaboration with other painters of his school), and similarly stress appears as a cause of peptic ulcer in two separate theories, but we are not in a position to deflate these different positions as compatible subsets or intersections of one another: the relevant theories are, from an epistemological point of view, separate and disjoint even when sharing individual assertions.

Several works on uncertainty in KGs confirm the centrality of this issue, pointing out that conflicts, heterogeneous sources, and competing interpretations of statement aggregates are pervasive, yet are usually managed procedurally rather than modeled explicitly \cite{couceiro2025uncertaintySurvey}. In \cite{vogt2024semanticUnits} the notion of \textit{semantic units} is introduced, namely meaningful subgraphs that support structured reasoning beyond the level of individual triples: epistemically relevant content in KGs needs to be organized at a higher level than single statements.

\bigskip
Yet, even these approaches stop short of providing an explicit epistemic characterization of such units. They identify useful aggregates, but do not specify how they relate to reality or to one another.
Indeed, past attempts for epistemic characterization of KGs raised many concerns about reasoning, witnessed by the well-known \textit{Superman Paradox}: from “Lois Lane believes that Superman can fly” one should not be allowed to infer that “Lois Lane believes that Clark Kent can fly,” even if it is a fact that Superman \textit{is indeed} Clark Kent. 

Such issues have historically discouraged the adoption of epistemic characterization and reasoning on quoted triples in the Semantic Web (e.g., see \cite{w3cAttributions2001}, \cite{rdfStarEditorsDraft}).
As a result, most KG approaches avoid assigning semantic force to epistemic predicates, treating them as opaque annotations rather than as operators with logical consequences. This avoidance, however, comes at the cost of expressive power and conceptual clarity.

\section{Cognitive modal logics} \label{sec:modallogics}

In this section we introduce the concepts of cognitive modal logics that are necessary to turn provenance annotations into well-defined epistemic structures over cognitive worlds, enabling controlled reasoning without the classical opacity pitfalls.

\subsection{Modal logics}\label{subsec:modal-logics}

Classical formal logics define truth as a binary property of formulas relative to a single interpretation. Each formula is either true or false, and logical operators (e.g., \textit{and, or, not, implies}) combine truth values according to fixed tables. Formally, modal logics extend a base logic with modal operators: although propositional logic is often used for exposition, the same construction applies to first-order logics and to description logics, which are particularly relevant in the Semantic Web.
The first additional operator, usually written as $\Box$ (“box”),  qualifies statements as belonging to the modality. The operator is defined formally by a structure called a \emph{Kripke frame}, and its semantics by an attribution of meaning to the box operator which qualifies and characterizes the logic itself in terms of the possible worlds in which the modal formula could become true (see Appendix~\ref{app:kripke} for a more formal introduction to Kripke frames and their properties in modal logic). A second additional operator, $\Diamond$ (“diamond”), is a complementary notion to $\Box$ which is frequently represented as $\Diamond \varphi = \neg \Box \neg \varphi$.

Different readings of the semantics of the $\Box$ and $\Diamond$ operators yield different modal logics, each with its own axioms and frame conditions.
\begin{itemize}
\item \textbf{Alethic logic} (necessity and possibility): 
Connections between worlds represent \emph{logical possibility}, so that $\Box\varphi$ means $\varphi$ is \textit{necessarily true}, while $\Diamond \varphi$ means $\varphi$ is \textit{possibly true}.
\item \textbf{Temporal logic} (time): Connections encode temporal progression. $\Box\varphi$ means $\varphi$ is \textit{always true}, $\Diamond\varphi$ means $\varphi$ is \textit{sometimes true}. 
\item \textbf{Deontic logic} (obligation and permission): Connections capture normatively ideal worlds. $\Box\varphi$ reads as \textit{it is obligatory} that $\varphi$, $\Diamond\varphi$ as \textit{it is permitted} that $\varphi$.
\item \textbf{Provability logic} (arithmetical provability): $\Box\varphi$ is interpreted as $\varphi$ is \textit{provable in a given formal theory $\mathit{T}$} (e.g., Peano Arithmetic), while  $\Diamond\varphi$ is interpreted as $\varphi$ is \textit{consistent with $\mathit{T}$}. 
\item \textbf{Cognitive logics} (subjective perspective): Cognitive logics are a family of modal logics, rather than a single one meant to represent a subjective point of view. They are often enriched with an index $i$ to represent the coexistence of multiple independent subjectivities associated with separate agents. A choice of axioms characterizes the nature and the interaction with reality of each cognitive world. These yield that in doxastic logics $\Box_i \varphi$ means that \textit{agent $i$ believes $\varphi$ is true}, in epistemic logics $\Box_i \varphi$ is used to represent that \textit{agent $i$ knows $\varphi$ is true}, and in conjectural logics $\Box_i \varphi$ reads as \textit{agent $i$ supposes $\varphi$ is true}. Cognitive modal logics have little use for the secondary operator $\Diamond \varphi$ so we will mostly ignore it. 
\end{itemize}

The formal aspects of the connections between possible worlds are characterized by one or more axioms that describe specific proof principles. In particular, an important family of modal logics is called \textit{normal}, and its axioms are well-established and well-studied. 

Specifically, a modal logic is called \emph{normal} if it is an axiomatic system that extends a base logic with a unary modal operator $\Box$ and satisfies the following proof principles.

\begin{itemize}
\item \textbf{Base logic} \\
{\small All valid formulas of the base logic are available as axioms, and the system is closed under Modus Ponens $((\varphi \rightarrow \psi) \wedge \varphi  \ \rightarrow \ \psi)$ and Uniform Substitution ($\varphi \rightarrow \varphi[\sigma]$, where $\sigma$ is a substitution replacing every variable with an arbitrary formula). The non-modal inferential structure is exactly as in the chosen base logic.} 
\item \textbf{Rule N (Necessitation)}: From $\vdash \varphi$ infer $\vdash \Box \varphi$. \\
{\small Rule \textbf{N} injects theorems of the base logic into the modal level, so anything derivable without $\Box$ can also be derived with it.}
\item \textbf{Axiom K}: $\Box(\varphi \rightarrow \psi) \rightarrow (\Box \varphi \rightarrow \Box \psi)$ \\ 
{\small Axiom \textbf{K} makes reasoning under $\Box$ support distributivity inside the modal logic, guaranteeing that reasoning mirrors Modus Ponens inside the modality.}
\end{itemize}

Additionally, they may be chosen to satisfy one or more of the following axioms: 
\begin{itemize}
\item \textbf{Axiom D}: $\Box \varphi \rightarrow \neg \Box \neg \varphi$ \\ 
{\small Axiom \textbf{D} excludes modal inconsistency, since the system cannot simultaneously validate $\Box \varphi$ and $\Box \neg \varphi$, and it also forbids trivialities such as $\Box \bot$.}
\item \textbf{Axiom T}: $\Box \varphi \rightarrow \varphi$ \\
{\small Axiom \textbf{T} enforces reflection from the modal level to the base level, so that anything established under $\Box$ must also hold without it.}
\item \textbf{Axiom 4}: $\Box \varphi \rightarrow \Box \Box \varphi$ \\
{\small Axiom \textbf{4} enforces stability under iteration: once a formula is inside $\Box$, reapplying $\Box$ adds no further strength to it.}
\item \textbf{Axiom 5}: $\neg \Box \varphi \rightarrow \Box \neg \Box \varphi$ \\ 
{\small Axiom \textbf{5} enforces recognizability of non-commitment: whenever $\Box \varphi$ fails, the system can state that it fails.}
\end{itemize}

An additional axiom, discussed in \cite{Vit26} and not part of the usual set of axioms of normal modal logics, is worth being introduced here:

\begin{itemize}
\item \textbf{Axiom C}: $\varphi \rightarrow \Box \varphi$ \\
{\small Axiom \textbf{C} expresses the idea that the modality mirrors at least the full set of defined formulas of the base logic.} 
\end{itemize}

In literature, frequent combinations of axioms are given customary names, such as  $\mathbf{K}$ for the core system with Base, Rule \textbf{N}, and axiom \textbf{K}; \textbf{KD45} for Base, Rule \textbf{N}, and axioms \textbf{K, D, 4,} and \textbf{5}; and \textbf{S5} for Base, Rule \textbf{N}, and axioms \textbf{K, T,} and \textbf{5}\footnote{Please note that \textbf{T} implies \textbf{D}, and \textbf{T} + \textbf{5} implies \textbf{4}: although unmentioned, they both hold in \textbf{S5}. }. We add to this list  at least two new combinations, namely \textbf{KC} and \textbf{KCD}, for Base, Rule \textbf{N}, and axioms \textbf{K, C, [D]}\footnote{\textbf{C} implies both \textbf{4} and \textbf{5}, so they hold in \textbf{KC} and \textbf{KDC} even if not explicitly mentioned.}.  

\bigskip

Beyond their formal role, the modal axioms introduced above admit a direct conceptual reading that is particularly relevant for cognitive interpretations, i.e., the main discussion topic of this paper.

\begin{itemize}

\item Rule \textbf{N} introduces an idea of \textit{predictability}: every theorem on the non-modal fragment of the logic is also a theorem inside the modal fragment. Thus reasoning in the modal world is understandable and homogeneous and predictable, being fully in line with reasoning in the outside world. This also implies that the conclusions you may reach inside the modality are only dependent on the initial atomic modal facts, and not on a different reasoning model.  

\item Axiom \textbf{K} expresses a principle of \emph{locality}: reasoning inside the modality is dependent only on what is already accessible under that modality. Inference under $\Box$ mirrors inference in the base logic, but remains confined to the modal context. This ensures that implications are available locally, without polluting distinct informational worldviews.

\item Axiom \textbf{D} captures a notion of \emph{rationality}. By excluding the simultaneous validation of $\Box\varphi$ and $\Box\neg\varphi$, it enforces internal coherence of the modal state. In cognitive terms, this corresponds to the requirement that an agent’s beliefs, assumptions, or conjectures do not allow contradictions.

\item Axioms \textbf{T} and \textbf{C} regulate different forms of \textit{permeation}, i.e., reciprocal access to formulas between the base logic and the modal layer. Axiom \textbf{T} enforces \textit{inside-out permeation}: whatever is established under the modality must also hold at the base level. Axiom \textbf{C}, by contrast, enforces \textit{outside-in permeation}: whatever is established at the base level is inherited by the modality. Depending on their presence, these axioms determine whether the modality isolates, constrains, or extends the formulas of the base logic.

\item Modal logics frequently ensure that the inferential mechanisms of the base logic, and most notably Uniform Substitution, do not traverse the locality boundaries of the modal operators. This principle is known as \textit{intensionality}, implying that formulas which are equivalent in the base logic are not, in general, interchangeable within modal contexts. This restriction, standard in modal logic, will play a central role in the treatment of cognitive modalities and will be discussed further in Section~\ref{sec:superman}.

\item If the choice of axioms includes both axiom \textbf{T} and axioms \textbf{C}, the corresponding modal logic immediately leads to the so-called \textit{modal collapse}, i.e., a system in which the modal and non-modal sets of formulas are identical, or to vacuity in the modal system: $\Box \varphi \rightarrow \varphi \text{ and }  \varphi \rightarrow \Box \varphi \text{ imply } \Box \varphi \iff \varphi  $.  

\item Modal collapse is caused by axiom \textbf{C} also without axiom \textbf{T}, as long as the base logic is \textit{bivalent}, namely, \textit{every well-formed formula is always determined}, i.e., either true or false (in addition to being nontrivial and classical). However, without bivalence (and without axiom \textbf{T}), such collapse is not inevitable, and non-trivial systems can be generated. Thus axiom \textbf{C} makes sense when using a non-bivalent base logic such as Supervaluation \cite{Fine1975Supervaluation}, Kleene's Weak Logic \cite{Kleene1952Metamathematics} or, relevant for our paper, Description Logics \cite{BaaderEtAl2003DLHandbook}, allowing for well-formed formulas to be either undefined or unentailed. A longer discussion of this axiom, and of the necessity of the assumption of non-bivalence, is explained in more detail in \cite{Vit26}. 

\end{itemize}

\subsection{Cognitive modal logics}\label{subsec:cognitive-modal-logics}

We call \emph{cognitive modal logics} a family of systems designed to model subjective informational states, that is, agents’ beliefs, knowledge, or hypotheses, rather than absolute necessity or permission or temporal features. 
Cognitive modal logics share some structural characteristics. First, their semantics explicitly separates reality from subjective states. Second, they admit multiple independent modal systems, one per agent or cognitive stance, each governed by its own relation and valuation. Finally, rather than treating additional axioms as incremental refinements of a single modality, cognitive logics treat axiom choices and combinations as the foundation of qualitatively distinct reasoning modes. Here we discuss three types of cognitive modal logics, \textit{doxastic, epistemic}, and \textit{conjectural}, hereinafter described as \textit{DEC modal logics}. 

\paragraph{Doxastic logic (KD45)}

 Intuitively, for each agent $i$, the operator $\Box_i$ formalizes that the agent takes a statement to hold within its beliefs regardless of its valuation outside of them. Axiom \textbf{D} enforces internal coherence of the belief state, axiom \textbf{4} enforces stability under iteration for what is inside the belief state, and axiom \textbf{5} enforces recognizability of non commitment. Together they yield the well-known system \textbf{KD45}, the canonical abstract account of idealized belief \cite{Hintikka1962KnowledgeBelief}.

\paragraph{Epistemic logic (S5)}

In its standard normal form, epistemic logic is captured by the combination of axioms \textbf{K}, \textbf{T}, and \textbf{5} over the chosen base logic with Rule \textbf{N} (note that \textbf{T} and \textbf{5} entail \textbf{D} and \textbf{4}). For each agent $i$, the operator $\Box_i$ now formalizes that the agent knows the statement, which presupposes its factuality. Additionally the choice of axioms also enforces internal coherence (axiom \textbf{D}, now available as a theorem, is a direct application of Axiom \textbf{T}), and the two standard forms of introspection. The resulting system \textbf{S5} is the classical abstract account of factual knowledge in multi-agent settings \cite{Fagin1995ReasoningKnowledge}.

\paragraph{Conjectural logic (KDC)}

Conjectural logic models hypothetical reasoning based on true information. Agents have access to all currently established facts, and in addition they may adopt one or more unverified assumptions (e.g., in order to explore their consequences). Formally, the standard axioms of rational doxastic logic, \textbf{K} and \textbf{D}, are enriched with axiom \textbf{C} that stipulates that all established predicates are adopted within the hypothetical scenario together with chosen non-established ones. Introspection, under the form of axioms \textbf{4} and \textbf{5}, is provided as theorems within this setup, rather than axioms. 

\subsubsection{The inclusion theorem}

The relationships between doxastic, epistemic, and conjectural logics can be understood within a precise structural characterization, expressed by the inclusion theorem of cognitive logics \cite{Vit26}, a unifying semantic account of how these logics relate to one another and to reality.

\begin{definition}[Cognitive world]\label{def:cognitive-world}
Let $P$ be a non-empty set of propositional atoms, and $\top$ and $\bot$ be two distinct values, interpreted as “true” and “false” respectively.
 A \emph{cognitive world} is a set $w \subseteq P \times \{\top,\bot\}$. Atoms not occurring in $w$ are left unassigned.
A world $w$ is \emph{consistent} iff there is no $p\in P$ such that $(p,\top)\in w$ and $(p,\bot)\in w$.
\end{definition}

\begin{definition}[Reality]\label{def:reality}
Reality is a designated cognitive world $r$ whose only requirement is to be consistent.
\end{definition}

\begin{definition}[Designated modal logic]\label{def:designated-modal-logic}
A \emph{designated modal logic} for $(r,w)$, written $\Sigma_{r,w}$, evaluates atomic formulas at $r$ as follows: 
\[
\begin{array}{r @{\;} c @{\;} l @{\;} c @{\;} l}
\Sigma_{r,w} & \models & p      & \iff & (p,\top)\in r, \\
\Sigma_{r,w} & \models & \Box q & \iff & (q,\top)\in w.
\end{array}
\]

The evaluation is then extended to all non-atomic formulas of the base logic by standard Boolean clauses and Rule~\textbf{N}.

\end{definition}

\begin{definition}[Cognitive classes]\label{def:cognitive-classes}
We define:
\begin{align*}
\textit{D}_r &:= \{\, w \subseteq P\times\{\top,\bot\} \,\} \\
\textit{E}_r &:= \{\, w \subseteq P\times\{\top,\bot\} \mid w \subseteq r \,\} \\
\textit{C}_r &:= \{\, w \subseteq P\times\{\top,\bot\} \mid r\subseteq w \,\}.
\end{align*}
\end{definition}

Each cognitive world $x$ induces, via $\Sigma_{r,x}$, a corresponding modal logic, and the inclusion theorem provides a strong set-theoretical connection between the cognitive classes of $r$ and the modal logics built upon $r$\footnote{We added index $r$ to $\textit{D}_r$ out of uniformity with the other sets, although strictly speaking it does not depend on $r$ at all.}:

\begin{theorem}[Inclusion Theorem]\label{thm:inclusion}
For any cognitive world $x \subseteq P\times\{\top,\bot\}$:
\begin{align*}
\Sigma_{r,x} \text{ validates } \mathbf{K}
    \;&\Longleftrightarrow\; x \in \textit{D}_r; \\
\Sigma_{r,x} \text{ validates } \mathbf{K}\mathbf{D}
    \;&\Longleftrightarrow\; x \in \textit{D}_r\ \text{and } x\ \text{is consistent,} & \text{i.e., x is doxastic}; \\
\Sigma_{r,x} \text{ validates } \mathbf{K}\mathbf{T}
    \;&\Longleftrightarrow\; x \in \textit{E}_r  & \text{i.e., x is epistemic}; \\
\Sigma_{r,x} \text{ validates } \mathbf{K}\mathbf{C}
    \;&\Longleftrightarrow\; x \in \textit{C}_r; \\
\Sigma_{r,x} \text{ validates } \mathbf{K}\mathbf{D}\mathbf{C}
    \;&\Longleftrightarrow\; x \in \textit{C}_r\ \text{and } x\ \text{is consistent},  & \text{i.e., x is conjectural}.
\end{align*}
\end{theorem}

The inclusion theorem therefore establishes a strict correspondence between cognitive modal axioms and the set-theoretic relations of cognitive worlds with respect to the designated reality: all consistent cognitive worlds built upon subsets of $P$ are doxastic, and once a reality world $r$ is designated, all subsets of $r$ are doxastic worlds, and all supersets of $r$ are conjectural worlds. 

\subsubsection{The \textit{settle} operator}

Dynamic modal logics extend static modal systems with operators that describe the evolution of epistemic states over time. In the epistemic setting, this tradition is represented by Dynamic Epistemic Logic (DEL) \cite{vanDitmarsch2007DEL}, whose operators—such as public announcement, update, and revision—model changes in information by restricting the set of admissible worlds. Formally, these operators act by eliminating worlds that are incompatible with newly acquired information, thereby refining an agent’s epistemic range.

Within the present framework, however, such operators are insufficient. The inclusion theorem characterizes cognitive logics not by accessibility relations among worlds, but by the set-theoretic position of a cognitive world $x$ with respect to a designated reality $r$. In this setting, epistemic change is not primarily a matter of discarding worlds, but of modifying which propositions count as established at the base level.

To account for this form of epistemic evolution, \cite{Vit26} introduces the dynamic operator \emph{settle}. The role of \textit{settle} is to transform conjectural commitment into factual establishment. 

Formally, the operator \textit{settle} is applicable to conjectural worlds, i.e., to $x$ such that $x \in \textit{C}_r$. 
Given a formula $\varphi$ such that $(\varphi,\top)\in x$ but $(\varphi,\top)\notin r$, the application of $\textit{settle}\,\varphi$ produces a new designated reality $r' = r \cup \{(\varphi,\top)\}$. 
With respect to $r'$, the old reality $r$ becomes epistemic, since $r \subseteq r'$ now holds.

Unlike standard DEL operators, \textit{settle} induces a shift in the inclusion hierarchy: a world that was conjectural with respect to the original reality becomes the new reality, makes the old reality an epistemic world, and makes delusional every conjectural world containing $(\varphi,\bot)$. 
The purpose of the \textit{settle} operator is therefore to capture epistemic stabilization without collapsing the modal structure or resorting to world elimination, providing a formal mechanism for transitions between cognitive worlds.

\section{The DEC Framework}\label{sec:DECframework} 

This section introduces how to apply cognitive modal logics to the Semantic Web to model provenance-enhanced content. We call this approach the \textit{DEC framework} (for \textit{doxastic, epistemic }and\textit{ conjectural}).

DEC takes as its primary target provenance-enhanced statements (PES), understood as statements whose presence in a knowledge graph is justified by an act of assertion rather than by immediate factual verification, i.e., as \textit{capta}: their semantic relevance depends on who asserted them, under which conditions, and with which epistemic commitment.
PES typically do not occur in isolation. Instead, they appear in structured aggregates (\textit{theories} or \textit{cognitive worlds}) that are proposed, defended, contested, or abandoned as a whole, even when some of their constituent statements may also occur in other, competing theories. 

Within DEC, theories are the primary carriers of epistemic characterisation: provenance identifies aggregates of claims to be interpreted together. The semantic role of provenance is therefore to delimit such units and to support their epistemic interpretation. This role is realized through governing predicates: provenance relations that link agents, sources, or justificatory acts to theories and therefore to their constituent statements. Governing predicates determine how a theory is to be interpreted epistemically and how it relates to other theories and to established knowledge.  

\subsection{Cognitive Worlds and Cognitive Modal Logics}\label{subsec:cw-and-cml}

As mentioned, \textit{cognitive worlds}, i.e., theories, are collections of statements that are homogeneous with respect to provenance and epistemic commitment, and that must be evaluated as a whole rather than statement by statement.

Cognitive modal logics, i.e., doxastic, epistemic, and conjectural logics, provide both the conceptual vocabulary and the mathematical structure for describing how agents relate to asserted facts, how they organise beliefs and hypotheses, and how these attitudes interact. In DEC, we adopt these logics and their axioms as a guide to classify and constrain different kinds of cognitive worlds. We shall use a simple zoological setting for our examples.  

At the base of all cognitive worlds lies axiom \textbf{K}, which enforces \textit{locality}: what is true or inferable inside a cognitive world depends only on the content of that world. No inference automatically crosses the boundary between worlds, and no commitment is inherited unless explicitly specified. This corresponds directly to the idea that a theory, understood as a provenance-homogeneous aggregation of statements, stands or falls on its own terms. 

From \textbf{K}, additional constraints generate more specific kinds of cognitive worlds. Adding axiom \textbf{D} introduces \textit{rationality} by ruling out internal inconsistency. A doxastic world governed by \textbf{K+D} is still local and belief-based, but it cannot contain explicit contradictions i.e., it is rational, even if its beliefs are at odds with asserted facts. For instance, it is a fact that whales are mammals, and that mammals and fishes are disjoint, but consider the case in which Bruce believes that whales are fishes. If we characterize Bruce's cognitive world as doxastic, facts and Bruce's belief are independently characterized and may coexist without permeation or contradiction. 

Adding axiom \textbf{T} introduces \textit{inside-out permeation}: every local statement  is also asserted (i.e., true). \textbf{K + T} therefore characterises epistemic worlds, whose content represents (possibly incomplete) \textit{knowledge of facts} and not mere belief. Epistemic worlds are constrained both by locality and by rationality: they are internally consistent, and they include all asserted facts relevant to their domain. Suppose for instance that Alice finds out that Moby Dick is a whale, and let us consider her beliefs as epistemic (i.e., facts). Regardless of Alice's own knowledge of biological facts about whales, this is enough to let us conclude that Moby Dick is, in fact, a mammal, since Alice's local statement permeated into reality. 

Starting again from \textbf{K}, axiom \textbf{C} introduces \textit{outside-in permeation}, injecting asserted facts into the cognitive world. \textbf{K + C} characterise conjectural worlds, in which hypotheses extend the domain of asserted facts. When rationality is also required, \textbf{K + D + C} characterise internally coherent hypothetical extensions of factuality. For instance, assume that Carl conjectures dragons to be reptiles. This hypothesis does not contradict asserted facts such that dragons are fictional animals and that Toothless is a dragon, and Carl's point of view therefore produces a coherent conjectural world. It follows that, within Carl's conjecture, it is possible to conclude that Toothless is fictional and that Toothless is a reptile. 

In addition to doxastic, epistemic, and conjectural worlds, DEC also recognises \textit{verbatim worlds}, which record statements without cognitive commitment, \textit{personal opinions}, completely detached from reasonably verifiable statements (as in "David believes that pizza is better than hot dogs"), and \textit{delusional worlds}, which prove inconsistent when confronted with asserted facts (in fact, Bruce's belief that whales are fishes is in direct contradiction with asserted facts, and therefore Bruce's cognitive world is both doxastic \textit{and} delusional). 

Regardless of the typology of cognitive world, the same structural principles apply: cognitive worlds are collections of statements that share provenance and epistemic fate; locality controls inference boundaries; rationality constrains internal coherence; permeation regulates the relationship between worlds and asserted facts. Together, these principles allow DEC to model a wide spectrum of epistemic situations while keeping the distinction between facts, beliefs, and hypotheses explicit and formally grounded.

The inclusion theorem allows us to conclude that, given a set of atomic statements $P$, and a selection of assigned statements within $P$, called $r$ (for \textit{reality}), every set of assigned statements is a doxastic cognitive world, any subset of $r$ is an epistemic cognitive world, and every superset of $r$ is a conjectural world. If we then consider $\overline{r}$ as the opposite of $r$ (a cognitive world with the same assignments as $r$, but where every true statement in $r$ is false and vice versa), then a \textit{delusion} is a doxastic world $w$ where $w \cap \overline{r} \neq \varnothing$ (i.e., there is at least one statement contradicting reality), while a \textit{personal opinion} is a cognitive world $y$ such that $y \cap r = \varnothing \ \land \ y \cap \overline{r} = \varnothing$, i.e., no statement in $y$ is asserted, either as true or as false. 

\subsection{Representing Cognitive Worlds with Semantic Web Technologies}\label{subsec:representation}

DEC does not prescribe a specific representation formalism. Instead, it specifies how epistemic commitment is to be interpreted once statements and their provenance have been encoded using existing Semantic Web technologies. This section illustrates how cognitive worlds can be represented in practice, distinguishing between the basic representation of statements, the addition of epistemic characterisation, and the derivation of factuality and inference.

At the base level, the asserted fact "\textit{Moby Dick is a whale}" is represented as an ordinary RDF triple. This representation carries no epistemic information: it expresses a statement that is taken to hold in the dataset, without specifying who asserts it or under which epistemic stance.

At this level, RDF semantics applies uniformly. If the triple belongs to the set of asserted facts, it contributes to factuality and participates fully in inference.

Epistemic characterisation is introduced by associating the statement with a governing actor (e.g., \textit{Alice}) via a governing predicate (e.g., \textit{says, believes, knows, supposes}). Conceptually, this creates a cognitive world.
\begin{itemize}
    \item Using \textbf{named graphs}, the triple is placed inside a graph that is linked to Alice via a cognitive predicate. The graph as a whole represents Alice’s cognitive world. The internal content is the statement \textit{Moby Dick is a whale}; the external relation expresses that Alice has a specific cognitive relation to it. For example: 
    \captionsetup{type=listing}
    \begin{Verbatim}
    
    :whale rdfs:subClassOf :mammal.
    :alice :knows :factsKnownByAlice.
    GRAPH :factsKnownByAlice {
        :mobyDick a :whale.
    }
    \end{Verbatim}
    \captionof{listing}{A cognitive world as a named graph.}
    \label{lst:rep-named-graphs}
    
    \item Using \textbf{RDF-star} (or, RDF 1.2), the same triple is embedded and annotated with a provenance relation indicating that Alice knows it. The cognitive world is not an explicit container but the set of embedded triples governed by the same epistemic predicate.
    \captionsetup{type=listing}
    \begin{Verbatim}

:whale rdfs:subClassOf :mammal.
:alice :knows << :mobyDick a :whale >>.
    \end{Verbatim}
    \captionof{listing}{A cognitive world as a quoted triple.}
    \label{lst:rep-rdf-star}

    \item Using classical \textbf{reification} or \textbf{n-ary relations}, the statement is represented as a node with explicit links to subject, predicate, and object, and that node is connected to Alice through an epistemic relation. The cognitive world is reconstructed as the collection of all statement nodes sharing the same governance.
    \captionsetup{type=listing}
    \begin{Verbatim}

:whale rdfs:subClassOf :mammal.
:alice :knows :s1.
:s1 a rdf:Statement;
    rdf:subject :mobyDick;
    rdf:predicate rdf:type;
    rdf:object :whale.
    \end{Verbatim}
    \captionof{listing}{A cognitive world as an RDF reification.}
    \label{lst:rep-reification}
    
\end{itemize}

So far, no cognitive characterisation has been specified, and in all cases, such characterisation is external to the statement content. What differs is how strongly the internal entities are isolated from the rest of the graph. 

As explained in section \ref{sec:decreasoner}, DEC uses two ways to characterise cognitive worlds and therefore their cognitive stance, either by explicitly assigning a DEC type to the world, e.g., 
\begin{Verbatim}

    :factsKnownByAlice a dec:epistemicWorld.
    
\end{Verbatim}
\label{lst:rep-explicit-world-typing}

, or by assigning a cognitive stance to the governing predicate, e.g.: 

\begin{Verbatim}

    :knows a dec:epistemicPredicate.
    
\end{Verbatim}
\label{lst:rep-predicate-typing}

Of course, the first approach is only possible if the cognitive world is provided with an identifier, i.e. with named graphs. Other syntaxes require the characterization of the predicate. 

If Alice’s cognitive world is typed as epistemic, as in this case, its content contributes to asserted facts. In this case, a DEC processor will \textit{permeate} its content in factuality and will make sure that it is treated as an asserted fact for semantic evaluation. If the same content were associated with Alice through a doxastic predicate, it would not permeate to factuality, regardless of internal representation. Similarly, if it were associated through a conjectural predicate, then all facts would permeate inside the cognitive world, regardless of internal representation.  

Thus all three representation models, when run through a DEC processor and a reasoner, would generate two more facts: 

\begin{Verbatim}

    :knows a dec:epistemicPredicate.
    :mobyDick a :whale, :mammal.
    
\end{Verbatim}
\label{lst:rep-permeation-consequences}

Thus DEC scaffolding is independent of syntax: contribution to factuality depends entirely on the epistemic typing of the cognitive world or of the governing predicate. Once cognitive worlds have been fully permeated, standard inference (RDF, OWL, etc.) applies separately on factuality and on each cognitive world. 

The chosen representation model does not affect cognitive inference but only how easily statements can interact between worlds and factuality. Reification and n-ary relations expose internal entities more directly. Even when a statement does not contribute to factuality, its components may still be accessible to inference. DEC does not eliminate this behaviour but makes it explicit: intensional isolation is a property of the representation choice, not of the epistemic model itself.

DEC assumes that all statements sharing the same governing actor and governing predicate belong to the same cognitive world, and if the chosen representation does not allow aggregation it will replicate the governing part accordingly. For instance, given:

\begin{Verbatim}

    :dragon rdfs:subClassOf :fictionalAnimal.
    :toothless a :dragon.
    :carl :supposes << :dragon rdfs:subClassOf :reptile >>.
    :supposes a dec:conjecturalPredicate.
    
\end{Verbatim}
\captionof{listing}{A conjecture as a quoted triple.}
\label{lst:rep-conjectural-input}

, where \texttt{:supposes} characterises conjectural worlds, i.e., permeates facts inside worlds, a DEC engine would be able to generate the following new quoted triples by repeating their governing data:

\begin{Verbatim}

    :carl :supposes << :dragon rdfs:subClassOf :fictionalAnimal >>.
    :carl :supposes << :toothless a :dragon >>.
    :carl :supposes << :toothless a :reptile >>.
    :carl :supposes << :toothless a :fictionalAnimal >>.
    
\end{Verbatim}
\captionof{listing}{Governing predicates are repeated for quoted triples.}
\label{lst:rep-conjectural-generated}

\subsection{DEC Interpretation and RDF Semantics}\label{subsec:dec-rdf-semantics}

DEC introduces a semantic structuring of RDF graphs in which some graphs are interpreted as expressing specific cognitive stances. Only a subset of these graphs contributes to semantic commitment.

RDF~1.2 defines interpretations for RDF graphs by assigning denotations to IRIs,
literals, and properties, and by specifying satisfaction and entailment for individual
graphs. Although RDF datasets may contain multiple named graphs, RDF~1.2 does not assign
any semantic role to graph names as dataset components: a graph name acquires a
denotation only when it occurs as an RDF term inside a triple. As a consequence, standard
RDF semantics applies to graphs in isolation and does not provide a semantic account of
datasets as structured collections of graphs.

DEC builds on this foundation. It leaves the RDF interpretation unchanged and introduces
its semantic extension at the level of graph organisation: graphs may be grouped,
related, and interpreted with respect to cognitive roles, while the meaning of RDF
terms and the evaluation of individual triples remain governed by RDF semantics.

Within DEC, graphs may be associated with \emph{cognitive worlds}. A cognitive world
collects statements that share provenance and cognitive status. This organisation is semantic rather than syntactic: it does not rely on
new RDF constructs, and it does not alter the interpretation of RDF terms. Indeed, the notion of cognitive worlds applies independently of the specific representational
mechanism used to associate statements with cognitive roles. Named graphs provide a
direct way to group statements into worlds, but equivalent groupings can also be
obtained through quoted triples, classical reification, or n-ary relations, by means of
governing relations that link statements to sources, agents, and cognitive roles. In
these cases, cognitive worlds are not explicit graph containers but emerge from the
structure expressing provenance relations. 

Across all representations, one basic principle
for entailment holds: only epistemic content contributes to factuality, and \textit{DEC satisfaction is determined by RDF satisfaction of epistemic content alone}.

In the following, let $\mathcal{G}$ denote the class of all RDF graphs.

\begin{definition}[DEC-interpretation]\label{def:dec-interpretation}
A \emph{DEC-interpretation} is a tuple
\[
J = (I, W, N, \iota, \gamma, \tau)
\]
where:
\begin{itemize}
\item $I$ is an RDF interpretation as defined in RDF~1.2;
\item $W$ is a non-empty set of cognitive worlds;
\item $N$ is a set of graph names designated as cognitive;
\item $\iota : N \rightarrow W$ is an injective mapping from cognitive graph names to cognitive worlds;
\item $\gamma : W \rightarrow \mathcal{G}$ assigns to each world its RDF graph content;
\item $\tau : W \rightarrow \mathcal{T}$ assigns a cognitive type to each world, and $\mathcal{T}$ is a finite set of cognitive types.
\end{itemize}
A graph $g_n \in \mathcal{G}$ is associated to its name $n \in N$ as follows:  $g_n = \gamma(\iota(n))$.

\end{definition}

No semantic role is assigned to graph names outside $N$, and the RDF interpretation $I$ is not extended to dataset components.

\bigskip

Among cognitive worlds, those characterized as \textsf{Epistemic} play a special role. DEC defines a graph called \emph{reality}, constructed from epistemic content.

\begin{definition}[Reality graph]\label{def:reality-graph}
Given a DEC-interpretation $J=(I,W,N,\iota,\gamma,\tau)$, the \emph{reality graph} is defined as:
\[
r = \bigcup \{\, \gamma(w) \mid w \in W \text{ and } \tau(w)=\textsf{Epistemic} \,\}.
\]
\end{definition}

 Since, intuitively, the nature of an epistemic graph is to represent a (possibly partial) slice of reality, even when it is not explicitly recognized as an independent world, we define \textbf{every subset of $r$ as epistemic}, regardless of whether it belongs to $W$ or it is explicitly labelled as epistemic through $\tau$. 

 The reality graph, and as a consequence all epistemic graphs, are assumed to be internally consistent.

\bigskip
DEC defines satisfaction as a relation between a DEC-interpretation and an arbitrary RDF graph.

\begin{definition}[DEC satisfaction]\label{def:dec-satisfaction}
Let $J=(I,W,N,\iota,\gamma,\tau)$ be a DEC-interpretation with reality graph $r$, and let $g$ be any RDF graph. We say that $J$ \emph{DEC-satisfies} $g$, written $J \models_{\mathrm{DEC}} g$, if:
\begin{itemize}
\item $I$ is an RDF interpretation for $g$; and
\item if $g \subseteq r$, then $I \models_{\mathrm{RDF}} g$.
\end{itemize}
\end{definition}

No further condition is imposed. In particular, graphs that are not subsets of $R$ impose no RDF truth conditions: only epistemic ones do. As an immediate consequence, DEC satisfaction collapses to RDF satisfaction of the sole reality graph.

\begin{theorem}\label{thm:dec-satisfaction-reduces-to-rdf}
Let $J=(I,W,N,\iota,\gamma,\tau)$ be a DEC-interpretation with reality graph $r$. Then
\[
    \left(\forall g \in \mathcal{G},\; J \models_{\mathrm{DEC}} g\right)
    \Longleftrightarrow
    I \models_{\mathrm{RDF}} r.
\]
\end{theorem}

\begin{proof} \
\begin{itemize}
    \item[$\Rightarrow$] Assume that $J \models_{\mathrm{DEC}} g$ for every RDF graph $g$ in $\mathcal{G}$. Since $r \subseteq r$, the definition of DEC satisfaction implies $I \models_{\mathrm{RDF}} r$.
    \item[$\Leftarrow$] Assume $I \models_{\mathrm{RDF}} r$. Let $g$ be any RDF graph in $\mathcal{G}$. By monotonicity of RDF satisfaction, an interpretation satisfying a graph also satisfies any of its subgraphs. Thus, if $I \models_{\mathrm{RDF}} r$ and $g \subseteq r$, we conclude that $I \models_{\mathrm{RDF}} g$ and therefore $J \models_{\mathrm{DEC}} g$. If $g \not\subseteq r$, DEC satisfaction holds vacuously.
\end{itemize}
\end{proof}

Thus, only epistemic worlds participate in the interpretation of $J$, and none of the other worlds contribute to it nor prevent it in any form. In fact, worlds that are not epistemic may be mutually incompatible, may conflict with reality, or may even be internally inconsistent, without affecting $r$ or the validity of a DEC-interpretation.

\subsection{Intensionality, Identity, and the Superman Paradox}\label{sec:superman}

A well-known difficulty in reasoning about statements enriched with contextual or provenance information is the so-called \textit{Superman paradox}. The paradox arises from
apparently innocuous identity claims, often exemplified by the well-known identity between Superman and Clark Kent, when combined with cognitive operators. The most frequent form of the paradox involves Superman/Clark Kent's love interest Lois Lane, who knows that Superman can fly but is unaware of the identity of the two personalities: 

    \captionsetup{type=listing}
    \begin{Verbatim}
    
    :superman owl:sameAs :clarkkent.
    :loislane :thinks << :superman :can :fly >>.}
    \end{Verbatim}
    \captionof{listing}{The usual Superman paradox.}
    \label{lst:superman-setup}

The paradox stems from the unacceptable inference derivable from these assertions, that 

    \begin{Verbatim}

    :loislane :thinks << :clarkkent :can :fly >>.

    \end{Verbatim}
    \label{lst:superman-unacceptable-inference}

Historically, the issue is rooted in the principle often attributed to Leibniz and
known as \textit{uniform substitution} or \textit{substitutivity of identicals}: if two terms are identical, then in any formula one may be substituted for the other without changing the formula's truth value. Thus any formula that contains a reference to \texttt{:superman} can be rewritten into another formula containing a reference to \texttt{:clarkkent}.  In classical extensional logics, this principle is treated as a global metalinguistic rule. 

Yet, modal logics, and in particular cognitive modal logics, depart from this setting. Truth is evaluated relative to worlds, and inference is constrained by modal operators. As a result, uniform substitution is no longer a global principle but an internal rule of the base logic, subject to the \textit{locality} enforced by the modal structure and specifically by axiom \textbf{K}: identity does not freely cross modal boundaries, and what is true or valid in one cognitive stance does not automatically propagate to others.

Without additional principles, such as explicit permeation rules, identity statements remain confined to the worlds in which they are asserted. The Superman paradox thus does not signal a failure of reasoning, but rather the necessity of distinguishing between extensional truth and intensional evaluation.

In Semantic Web discussions, concerns about the right choice between extensionality and intensionality have historically motivated caution toward reasoning over quoted or reified statements, and have influenced the limited uptake of full reasoning in approaches such as classical reification and, more recently, embedded triples. The fear is that allowing reasoning over statements about statements would reintroduce such types of paradoxes of substitution and identity. 

As implied in the previous example, in OWL identity is expressed through the property \texttt{owl:sameAs}. Yet, an analogous and symmetric example could be devised using \texttt{rdfs:subClassOf} instead\footnote{According to Wikidata, \\
\texttt{wd:Q188784 wdt:P279 wd:Q7643449.    \ \ \# superhero subclass of superhuman \\
wd:Q7643449 wdt:P279 wd:Q15632617. \# superhuman subclass of fictional human}
}, yielding:

    \captionsetup{type=listing}
    \begin{Verbatim}
    
    :superhero rdfs:subClassOf :fictionalHuman.
    :loislane :thinks << :superman a :superhero >>.
    \end{Verbatim}
    \captionof{listing}{A different Superman paradox.}
    \label{lst:superman-subclass-setup}

, which would never be proposed to justify

    \begin{Verbatim}

    :loislane :thinks << :superman a :fictionalHuman >>.

    \end{Verbatim}
    \label{lst:superman-unjustified-inference}
        
In OWL, \texttt{owl:sameAs} is an ontological predicate, not a meta-rewriting principle that can be applied outside of ontological reasoning: its effect is a plain consequence of the chosen semantics \textit{within OWL}. This design choice places the predicate inside the base logic (OWL, indeed), and once a modal or provenance-aware semantics is adopted, there is no requirement for the effects of \texttt{owl:sameAs} to be global or context-free. On the contrary, treating it intensionally is fully consistent with its status in OWL as an ontological predicate, in the same way that \texttt{rdfs:subClassOf} is an ontological predicate in both RDFS and OWL. 

In short, there are no differences in the permeation power of \texttt{owl:sameAs}, \texttt{rdfs:subClassOf} or any other ontological predicate of RDFS or OWL. Adopting intensional reasoning in cognitive modal logics does not weaken identity, nor does it undermine Semantic Web reasoning, but supports a strict locality paradigm: identity assertions, like all other statements, are evaluated relative to the cognitive worlds they appear in. Whether and how they propagate to reality or to other worlds depends only on the explicit permeation axiom in force. The
intensional treatment of identity thus provides a principled response to the Superman paradox and clarifies why provenance-aware reasoning can be both expressive and formally well-behaved.

\subsection{Settlements, disagreements, delusions, and opinions} \label{subsec:disagreements}
In addition to epistemically characterized cognitive worlds, we outline here a few relevant situations that become representable once provenance-homogeneous PESs are treated as cognitive worlds: \emph{settlements}, \emph{disagreements}, \emph{delusions}, and \emph{personal opinions}. With DEC these are meant as emergent phenomena definable in terms of locality, permeation, and consistency with respect to the reality graph $r$.

\paragraph{Settlements.}
As introduced in Section~\ref{sec:modallogics}, the dynamic operator \emph{settle} captures the transition from conjectural commitment to factual establishment. In the RDF setting this corresponds to promoting selected statements from a conjectural world into the epistemic layer (and therefore into $r$), typically as the result of external validation, consensus, or institutional acceptance. Importantly, settlement is not a data-cleaning operation that deletes competing capta: it is a change in their epistemic status. The conjectural world remains available (as a record of the hypothesis and of its provenance), while its settled fragment becomes part of the asserted factual core and participates in standard entailment.

\paragraph{Disagreements.}
DEC makes it possible to represent mutually incompatible theories without forcing the dataset itself to become inconsistent, because DEC satisfaction depends only on $r$. Two cognitive worlds $w_1$ and $w_2$ are said to be in \emph{disagreement} when their union cannot be jointly maintained under the chosen underlying entailment regime (e.g., OWL, RDFS, or a selected rule set). Intuitively, disagreement means that the two theories cannot both be correct at the same time, even if neither is asserted as factual. This provides two benefits: (i) incompatible claims can coexist in the same dataset without triggering trivialization of inference, and (ii) disagreement can be \emph{discovered} rather than declared, by checking which combinations of worlds yield incoherence. Uncertainty can be treated as an emergent property of a dataset, that is made evident by the structure of incompatibilities among available cognitive worlds.

\paragraph{Delusions.}
When disagreement involves $r$ itself, the conflict is stronger: it is not merely that two theories are incompatible, but that a world is incompatible with what the dataset currently treats as established. We call this situation a \emph{delusion}. Formally, a world is delusional (with respect to $r$) when adding its content to $r$ yields inconsistency. This notion matches the common intuition that some cognitive positions are not just alternative hypotheses but are, given what is already established, no longer acceptable.

Delusion also supports a key aspect of provenance-enhanced statements: the truth of the attribution and the truth of the attributed content are independent. As mentioned, the statement ``\textit{According to Archbishop Ussher, the Earth was created in 4004 BC}'' is a correct attribution of a false statement; the embedded claim is delusional with respect to $r$ (i.e., considered false) while the attribution is actually true. 

\paragraph{Personal opinions.}
DEC distinguishes \emph{personal opinions} from both disagreements and delusions. Personal opinions are cognitive worlds whose content is, by design, detached from any reasonably checkable reality: they are claims for which $r$ is meant to provide neither support nor refutation.
Typical cases are preferences, tastes, and private evaluations (e.g., ``\textit{David believes that pizza is better than hot dogs}'').
Such statements are not \textit{uncertainties} waiting for settlement one way or another, but evaluations whose commitments are personal or community-specific. In other words, they may be perfectly meaningful and comparable across agents, but not with reality or the external world.

DEC treats personal-opinion worlds as \emph{irreconcilable} with the factual core. Without this distinction, worlds that fail to be supported by $r$ could be confused with delusional ones.
In fact, incompatibilities between opinions are not, by default, to be interpreted as epistemic conflicts, and they only become so when we explicitly decide to compare such predicates with the established reality.

Two incompatible preferences may create a disagreement but not a delusion (with respect to $r$):
\begin{listing}[H]
\begin{Verbatim}
    :alice :believes << :pizza :betterThan :hotdogs >>.
    :bob   :believes << :hotdogs :betterThan :pizza >>.

    :believes a dec:doxasticPredicate.
\end{Verbatim}
\caption{Personal opinions are not delusions.}
\label{lst:opinions-preferences}
\end{listing}
The two worlds can be compared, clustered, or analysed sociologically, but there is no expectation that a DEC processor should try to ``repair'' them against $r$, nor that their incompatibility should count as a delusional conflict.

\section{A reasoner for cognitive worlds} \label{sec:decreasoner}

While modal logics provide the theoretical foundation for reasoning about subjectivity, they do not replace the existing inferential structures already defined within the Semantic Web stack. RDF, RDFS, and OWL constitute a well-established hierarchy of entailment regimes, each grounded in model-theoretic semantics. These regimes allow machines to infer subclass relations, property hierarchies, transitive connections, domain and range implications, and class memberships, all under the assumption of a single coherent world of discourse.

In this section we show how to embed it within a richer cognitive framework. Cognitive worlds exist, are identified and retain their own internal consistency, and continue to rely on standard RDFS/OWL reasoning for the local derivation of facts. Cognitive modal logics act as a meta-reasoning layer above traditional RDF inference: agents reason over structured knowledge expressed in OWL and RDFS, but the truth of what they assert is modally qualified. Our DEC Reasoner exploits this separation by performing standard entailment inside each graph using existing reasoning models, while a distinct component interprets the governing predicates to manage inter-world accessibility, coherence, and disagreement.

The DEC framework has been implemented as a prototype reasoner for RDF datasets.\footnote{The demonstrator is available at \url{https://w3id.org/conjectures/dec}. The DEC reasoner source code is available at \url{https://w3id.org/conjectures/dec/code/reasoner}, and the DEC viewer source code at \url{https://w3id.org/conjectures/dec/code/viewer}.} The implementation is built on Apache Jena Fuseki 5 and exposes an ordinary SPARQL endpoint.

The DEC reasoner recognises cognitive worlds, applies the appropriate permeation policies, materialises the resulting graphs, and delegates ordinary RDF, RDFS, or OWL inference to the configured backend reasoner. The reasoner acts at a level that is orthogonal to domain inference. 

It does not introduce new interpretations for RDF terms, nor does it assign special meaning to ontological predicates such as \texttt{rdfs:subClassOf}, \texttt{owl:sameAs}. These predicates are still interpreted by the underlying Semantic Web stack, while the DEC reasoner only determines where a statement is available, whether it contributes to reality, and whether its consequences remain confined to a cognitive world. Furthermore, it introduces a handful of domain-specific properties for DEC-related services. 

The reasoner is activated at dataset level. When DEC reasoning is inactive, the dataset behaves ordinarily. When it has been activated, all insertions and updates trigger a special DEC processing phase before the dataset is exposed to query answering.

\subsection{Cognitive worlds and predicates}
When inserting new triples the reasoner first identifies the cognitive worlds contained in the dataset and assigns a DEC type to each world (either from an explicit declaration or from a governing predicate). When a query is evaluated, it applies the permeation rules associated with the detected world types duplicating triples as specified in the cognitive type of the world. The reasoner then lets the configured RDF, RDFS, or OWL reasoner work over the permeated graphs as usual. Thus, DEC does not replace the underlying reasoner, but simply prepares the graphs on which the underlying reasoner must operate, permeating triples as needed.

The implementation recognises cognitive worlds in three RDF representation styles: named graphs, RDF-star quoted triples, and classical reification. The syntactic form changes how worlds are recovered from the dataset, while their DEC behaviour depends just on the epistemic typing.

With named graphs, the graph itself may be typed directly\footnote{Here and in the following the \texttt{dec:} prefix is associated to the URI \texttt{http://w3id.org/conjectures/}}:

\begin{listing}[H]
\begin{Verbatim}
    :alice :knows :MobyDickFacts .
    GRAPH :MobyDickFacts {
        :mobyDick rdf:type :Whale .
    }
    :MobyDickFacts rdf:type dec:epistemicWorld .
    :Whale rdfs:subClassOf :Mammal .
    
\end{Verbatim}
\caption{An epistemic world as a named graph.}
\label{lst:epistemic-world-named}
\end{listing}

The same world type may also be induced by characterizing a governing predicate:

\begin{listing}[H]
\begin{Verbatim}
    :alice :knows << :mobyDick rdf:type :Whale >>
        
    :knows rdf:type dec:epistemicPredicate .
    :Whale rdfs:subClassOf :Mammal .
    
\end{Verbatim}
\caption{An epistemic world induced by a direct governing predicate over a quoted triple .}
\label{lst:epistemic-world-direct}
\end{listing}

In this second case, an epistemic world is implicitly generated because \texttt{:knows} has been designed as an epistemic predicate. This is particularly useful to generate cognitive worlds containing sets of statements using quoted triples or reifications, as in this example. The mechanism is useful when the same predicate is repeatedly used to introduce worlds of the same kind. Indeed, when RDF-star or reification are used, the reasoner reconstructs worlds from governing predicates and agents only.

Reverse characterization for predicates having the quoted triple as the subject is also provided. This is useful for vocabularies such as PROV-O \cite{prov-o}, where some attribution patterns are naturally expressed from the statement, graph, or derivation towards the source. DEC supports reverse cognitive predicates without changing the underlying epistemic classification:

\begin{listing}[H]
\begin{Verbatim}
<< :mobyDick rdf:type :Whale >> prov:wasAttributedTo :melville. 
        
prov:wasAttributedTo rdf:type dec:reverseEpistemicPredicate .
:Whale rdfs:subClassOf :Mammal .

\end{Verbatim}
\caption{An epistemic world induced by a reverse governing predicate over a quoted triple .}
\label{lst:epistemic-world-reverse}
\end{listing}

The DEC reasoner supports a fixed set of DEC world and predicate types, summarised in Table~\ref{tab:cognitive-worlds}.

\begin{table}[ht]
\centering
\small
\renewcommand{\arraystretch}{1.5}
\begin{tabular}{@{}%
	p{0.36\linewidth}
	p{0.13\linewidth}
	p{0.46\linewidth}
@{}}
\hline
\textbf{Available classes} & \textbf{DEC type} & \textbf{Description} \\
\hline
\texttt{dec:verbatimWorld}\newline
\texttt{dec:verbatimPredicate}\newline
\texttt{dec:reverseVerbatimPredicate}
&
\textit{verbatim world}
&
Isolated from the other worlds and no reasoning takes place inside them. \\
\hline
\texttt{dec:sharedWorld}
\texttt{dec:sharedPredicate}\newline
\texttt{dec:reverseSharedPredicate}
&
\textit{shared world}
&
Contains in a single locations triples that are part of every cognitive world of the DEC universe (except \textit{verbatim worlds}). \\
\hline
\texttt{dec:doxasticWorld}\newline
\texttt{dec:doxasticPredicate}\newline
\texttt{dec:reverseDoxasticPredicate}
&
\textit{doxastic world}
&
Isolated from the other worlds but reasoning takes place locally.\\
\hline
\texttt{dec:epistemicWorld}\newline
\texttt{dec:epistemicPredicate}\newline
\texttt{dec:reverseEpistemicPredicate}
&
\textit{epistemic world}
&
Isolated from the other worlds but their statements permeate into reality. Reasoning takes place separately inside them and in reality.\\
\hline
\texttt{dec:conjecturalWorld}\newline
\texttt{dec:conjecturalPredicate}\newline
\texttt{dec:reverseConjecturalPredicate}
&
\textit{conjectural world}
&
Isolated from the other worlds but facts permeate into them. Reasoning takes place separately inside them.\\
\hline
\texttt{dec:delusionalWorld}
&
\textit{delusional world}
&
Contains at least one statement directly contradicting an asserted fact. This is usually an emergent classification after a consistency check.\\
\hline
\end{tabular}
\caption{Cognitive world and predicate types supported by the DEC reasoner.\label{tab:cognitive-worlds}}
\end{table}

 If no explicit DEC typing is given, the implementation applies conservative defaults in order to simulate the behavior of non-DEC reasoners. For classical reification and RDF 1.2 quoted triples, the default type is a plain verbatim world, blocking both permeation and reasoning on reified and quoted triples, as per \cite{w3cAttributions2001, CLK26}.
For named graphs, it is well known that RDF does not impose semantics: graph names are syntactically paired with graphs, and RDF semantics does not require them to denote the graphs they identify \cite{Zim14}. The DEC default therefore follows a common operational approach used by several RDF/OWL stores: the default DEC type for named graph treats it as an isolated graph with local inference (i.e., a doxastic world), while the so-called \textit{union graph} behavior offered by some RDF stores can be partially simulated by pairing unspecified named graphs with epistemic worlds. In this case their contents and consequences become factual, but while DEC preserves the epistemic origin of the materialised statements, union graph reasoning collapses generated statements in a single graph, usually the default graph.

\bigskip

The main operational characterization of a DEC world type is produced by its permeation behaviour. Two special cognitive worlds are relevant for determining the content and the permeation of the other worlds in the dataset: the shared world and the reality world. 

The shared world is defined as the union of all worlds defined as \texttt{dec:sharedWorld}, plus all the worlds referenced by governing predicates of type \texttt{dec:sharedPredicate} for direct references and \texttt{dec:reverseSharedPredicate} for reverse ones. All non-verbatim worlds permeate from the shared world. The shared world is meant as a utility structure containing background statements expressed just once but available inside all non-verbatim worlds and reality. Shared worlds are useful for taxonomies, common individuals, and modelling assumptions that would otherwise have to be repeated in each cognitive world of the universe.

The reality world is defined as the union of all worlds defined as \texttt{dec:epistemicWorld}, plus all the worlds referenced by governing predicates of type \texttt{dec:epistemicPredicate} for direct references and \texttt{dec:reverseEpistemicPredicate} for reverse ones. The default graph belongs to reality inasmuch as it is defined as an epistemic world itself (which is in fact the default). Also in reality is the shared world, permeated here as well as in all other non-verbatim worlds. 
Reality is therefore represented by a distinguished graph produced from ordinary factual assertions (the default graph), shared content, and the contribution of epistemic worlds. The DEC reasoner refers to the constructed reality graph with the standard identifier \texttt{dec:reality}.

Verbatim worlds are strict containers of assertions. They preserve the attributed content exactly as reported: no inference is performed inside them, and no permeation connects them with reality, with other worlds or even with shared knowledge. This makes verbatim worlds appropriate for exact quotations, inscriptions, documentary statements, or any case where the system must record exactly what was said or written without deriving consequences.

Doxastic, epistemic and conjectural worlds receive permeation according to a strict set-theoretical approach: doxastic and epistemic worlds receive shared triples; shared and epistemic worlds are used to generate the reality world; conjectural worlds receive reality triples, and their content does not permeate into reality. 

During the input of new quads, the reasoner keeps track and updates the declared worlds, governing predicates and associated DEC types. 
The current implementation materialises triples across permeating worlds during queries and then lets the base reasoner create the new knowledge separately in each cognitive world. So, for instance, any of the listings \ref{lst:epistemic-world-named}, \ref{lst:epistemic-world-direct} or \ref{lst:epistemic-world-reverse} generates the following new facts once the permeated worlds are subjected to a standard RDFS or OWL reasoning engine: 

\begin{listing}[H]
\begin{Verbatim}
:mobyDick rdf:type :Mammal.
:mobyDick rdf:type :Whale.

\end{Verbatim}
\caption{Permeation and reasoning from an epistemic world.}
\label{lst:permeation-epistemic-example}
\end{listing}

Updates are handled incrementally: the reasoner updates the affected worlds instead of recomputing all consequences at query time. This choice makes the endpoint easy to query and inspect, since SPARQL sees the resulting triples directly. It also exposes an implementation cost, especially with large numbers of worlds and especially of conjectural worlds: the same shared and permeated statements will be repeated in many worlds, increasing the computational costs of deriving new triples.

\subsection{Disagreements and delusions}\label{subsec:reasoner-conflicts}

The DEC reasoner can also detect inconsistencies between worlds (\textit{disagreements}) and between a world and reality (\textit{delusions}). This detection is optional, because the relevant checks may be expensive and because many datasets only need permeation and local reasoning. Thus, it is performed when requested by explicit DEC instructions in the dataset.

Our implementation provides two mechanisms for this detection, \texttt{dec:checkInconsistencies}, which works through the underlying reasoning engine, and \texttt{dec:closedFor}, handled directly by the DEC reasoner. They can be both activated at the same time, if needed. 

\bigskip
The \texttt{dec:checkInconsistencies} boolean predicate activates inconsistency checking through the configured underlying reasoner (in our case, the Fuseki OWL reasoner):

\captionsetup{type=listing}
\begin{Verbatim}

GRAPH :shared { 
    :mobyDick rdf:type :Whale . 
    :Mammal owl:disjointWith :Fish.
} 

GRAPH :aliceKnows { 
    :Whale rdfs:subClassOf :Mammal . 
}
:alice :knows :aliceKnows .

GRAPH :bruceBelieves { 
    :Whale rdfs:subClassOf :Fish . 
} 
:bruce :believes :bruceBelieves .

:shared rdf:type dec:sharedWorld .
:knows rdf:type dec:epistemicPredicate .
:believes rdf:type dec:doxasticPredicate .
[] dec:checkInconsistencies true.

\end{Verbatim}
\captionof{listing}{A request for inconsistency check.}
\label{lst:checkinconsistencies-input}

With this option, the DEC reasoner constructs relevant combinations of worlds and asks the underlying base reasoner whether their union is satisfiable. If two worlds cannot be jointly satisfied, the reasoner materialises a \texttt{dec:disagreesWith} relation between them. If a world disagrees with any world participating in reality (e.g., epistemic worlds and/or the default graph), the world is also typed as \texttt{dec:delusionalWorld}. 

For instance, listing \ref{lst:checkinconsistencies-input} specifies (as shared knowledge) that Moby Dick is a whale and mammals are disjoint from fishes, that (epistemically) Alice  knows that whales are mammals, and that (doxastically) Bruce believes that whales are fishes. 
This generates, among others, the following inferences: 

\captionsetup{type=listing}
\begin{Verbatim}

GRAPH :aliceKnows { 
    :mobyDick rdf:type :Mammal. 
}
GRAPH :bruceBelieves { 
    :mobyDick rdf:type :Fish. 
}

:mobyDick rdf:type :Whale.
:mobyDick rdf:type :Mammal.

:bruceBelieves dec:disagreesWith :aliceKnows.
:bruceBelieves dec:disagreesWith dec:reality.
:aliceKnows rdf:type dec:epistemicWorld.
:bruceBelieves rdf:type dec:doxasticWorld.
:bruceBelieves rdf:type dec:delusionalWorld.

\end{Verbatim}
\captionof{listing}{Inferred consequences of Listing~\ref{lst:checkinconsistencies-input}.}
\label{lst:checkinconsistencies-output}

The \texttt{dec:checkInconsistencies} mechanism depends on the configured inference regime, which in this case is OWL. While this approach is semantically general,  it still requires conflicts to be expressible as logical inconsistencies in the hosting ontology, i.e., it requires the dataset to contain enough explicit background knowledge to let these inconsistencies emerge: in this example, it is necessary to point out that mammal and fish classes are disjoint. 

\bigskip
The second mechanism, \texttt{dec:closedFor}, is specific to DEC and is based on explicit closure constraints. It addresses cases where disagreement is epistemically meaningful even if ordinary RDF or OWL reasoning would be too complicated or unable to produce an inconsistency. 

Suppose for instance we need to represent multiple theories about the attribution of a work of art, as in listing \ref{lst:hamlet}: 

\begin{listing}[ht]
\begin{Verbatim}
GRAPH :ShakespeareWroteHamlet { 
    :Hamlet dc:creator :Shakespeare. 
}
:scholars :agree :ShakespeareWroteHamlet .

GRAPH :DeVereWroteHamlet { 
    :Hamlet dc:creator :EdwardDeVere. 
}
:ThomasLooney :proposes :DeVereWroteHamlet .

:agree rdf:type dec:epistemicPredicate .
:proposes rdf:type dec:doxasticPredicate .

rdf:subject dec:closedFor dc:creator .
\end{Verbatim}
\caption{Closure-based disagreement on authorship (input dataset).}
\label{lst:hamlet}
\end{listing}

Under the open world assumption, the two worlds do not really disagree, since by considering the two graphs together we would not be able to generate an inconsistency in RDFS or in OWL, as the two \texttt{dc:creator} statements could simply be providing complementary information: e.g. we could deduce that both Shakespeare and De Vere collaborated in writing Hamlet, or that Shakespeare and De Vere were actually the same person. Yet, for our specific modelling purpose, we need to be able to treat separate author assignments as mutually incompatible.
DEC expresses this commitment through the \texttt{dec:closedFor} predicate over specific domain properties. 

A declaration such as \texttt{rdf:subject dec:closedFor dc:creator} says that, for the predicate \texttt{dc:creator}, worlds assigning different objects to the same subject are incompatible: each world is closed to multiple \texttt{dc:creator} assignments over the same subjects. Such closure is local to the specified predicate and to subject/object alternatives, so this assignment neither changes the global semantics of \texttt{dc:creator}, nor, importantly, does it make RDF closed-world in general.

Without resorting to the host reasoner, thus, the DEC reasoner is able to generate a few new triples: 

\begin{listing}[ht]
\begin{Verbatim}
:DeVereWroteHamlet dec:disagreesWith :ShakespeareWroteHamlet .
:DeVereWroteHamlet dec:disagreesWith dec:reality .

:ShakespeareWroteHamlet rdf:type dec:epistemicWorld .
:DeVereWroteHamlet rdf:type dec:doxasticWorld .
:DeVereWroteHamlet rdf:type dec:delusionalWorld .
\end{Verbatim}
\caption{Closure-based disagreement on authorship   .}
\label{lst:closedfor-authorship-output}
\end{listing}

The graph \texttt{:ShakespeareWroteHamlet} is epistemic because it is governed by \texttt{:agree}; it therefore contributes its authorship assignment to reality. The graph \texttt{:DeVereWroteHamlet} is doxastic because it is governed by \texttt{:proposes}. Since \texttt{dc:creator} has been closed for subject/object alternatives, the two worlds disagree, as they assign different creators to the same work. But since the Shakespeare attribution is epistemic, it permeated reality and therefore the De Vere attribution also disagrees with reality and is thus marked as delusional.

Yet, importantly, the DEC Reasoner does not erase the De Vere attribution, which is then preserved as a belief with explicit metadata describing its relation to the factual layer. A query can therefore retrieve the historical attribution, its proponent, its cognitive status, and its incompatibility with the settled attribution. This is precisely the kind of situation for which mere provenance is too weak: with DEC not only all attributions remain part of the dataset, but their epistemic positions become queryable.

\bigskip

As shown by our prototype, the DEC model can be implemented as a practical layer over existing Semantic Web infrastructure. Its main operation is graph construction: to recognise epistemically typed worlds, to materialise the permeations required by their type, to invoke the configured reasoner over the resulting graphs, and to record optional conflict annotations when requested. The semantics of DEC does not really require complete materialisation:  future implementation could compute consequences lazily and locally, and cache them more selectively. The present prototype materialises all permeated triples upon executing queries, which makes its behaviour inspectable, but at the cost of multiplying the same triples when many worlds share the same background. 

\section{Related works} \label{sec:relatedworks}

Handling provenance-enriched statements, with different approaches and different foci, has a rich and long history in the literature of the Semantic Web and neighboring fields. We specifically examine three topics: logical models for contextual reasoning, structural mechanisms for representing statements with provenance, and models to associate epistemic characterization to statements. Our main interest in this evaluation lies specifically in verifying the relation between an assertion as represented in the dataset and its factual commitment.

We first examine how contextualized knowledge is considered as part of the logical setting in which inference takes place. These approaches are primarily concerned with enabling reasoning over heterogeneous knowledge bases or datasets whose contents are locally organized and context-dependent. Their focus is to preserve local semantics (assumed as valid) and regulate interoperability across sources. Early contextual approaches to ontologies, such as C-OWL, distinguish shared ontological models from local contextual views and connect them by explicit context mappings \cite{BouquetGiunchigliaVanHarmelenSerafiniStuckenschmidt2003COWL}, while in \cite{JosephKuperMossakowskiSerafini2016BridgeRules} contextualized RDF and OWL knowledge bases use explicit bridge rules to reason across locally separated bodies of knowledge.
Others consider contexts themselves as first-class objects in contextualized knowledge. For instance, in \cite{KlarmanGutierrezBasulto2011TwoDimensionalDL} a two-dimensional description logic distinguishes the ordinary object domain from the domain of viewpoints or contexts. OWLC \cite{AljalboutBuchsFalquet2019OWLC} similarly separates context-dependent knowledge from knowledge about contexts and extends OWL-style reasoning with contextual rules. Contexts are seen as semantic objects rather than simply storage partitions, and afford their own context-oriented reasoning.
Finally, annotation-oriented Semantic Web frameworks attach annotations \cite{ZimmermannLopesPolleresStraccia2012AnnotatedSemanticWebData}, or "colors" \cite{FlourisFundulakiPediaditisTheoharisChristophides2009ColoringRDF}, to statements and define how such values propagate through inference: derived statements inherit, combine, or transform the context of their premises, but such annotations are still considered as neutral sources or abstract annotation values rather than epistemically precise characterizations of the statements. 

Next, we concern ourselves with the representational mechanisms used to express statements with provenance. Classical RDF reification and related meta-statement approaches make statements available as objects of discourse, but at the cost of verbosity and unclear inferential behavior \cite{YangKifer2003AnonymousResourcesMetaStatements}. N-ary relation patterns introduce intermediate entities to represent a complex relation including qualifiers, temporal annotations, provenance, or certainty \cite{HayesWelty2006NaryRelations}. Four-dimensional and fluent-based approaches also multiply the entities involved in a relation, introducing e.g., time slices, contextual parts, or relators so as to preserve ordinary predicates within context-dependent truth \cite{WeltyFikesMakarios2006ReusableFluents,BatsakisPetrakis2009TemporalRepresentationOWL,ZamborliniGuizzardi2013TemporallyChangingOWL}. NdFluents generalizes this strategy from time to multiple contextual dimensions \cite{GimenezGarciaZimmermannMaret2017NdFluents}. Singleton properties \cite{NguyenSheth2016LogicalInferencesContextsRDF} multiply predicates instead, creating numerous contextualized triples using multiple context-appropriate predicates over the same entities.  Yet, these approaches are not meant for characterizing the epistemic status of claims, and furthermore create numerous additional entities or predicates for different contextual truths to coexist in the same dataset. Named graphs offer a graph-level alternative for contexts by allowing triples to be grouped, named, separated, cited, imported, and associated with different provenance \cite{CarrollBizerHayesStickler2005NamedGraphs}; yet, named graphs do not have a single natural semantics, and may act as partitions, quoted graphs, isolated contexts, online graphs, or provenance carriers within the same representation approach \cite{ArndtVanWoensel2019MultipleSemanticsNamedGraphs}.
RDF-star and SPARQL-star provide the most recent mechanism for statement-level metadata: quoted triples may now occur directly as subjects or objects of other triples, making it easier to attach provenance, time, certainty, or source information to individual statements \cite{HCKS2021}. RDF officially adopted this approach in its most recent version, 1.2 \cite{CLK26}. In \cite{RuppSchnabelEckert2022MetadataModelingRDFStarNamedGraphs}, both RDF-star and named graphs have been used to distinguish between statement-level and graph-level metadata, while \cite{dibowski2024traceability, AmmannAlassi2024JourneyStar} use RDF-star for triple-level traceability and change provenance, showing the practicality of quoted triples for statement-level attribution. 

Finally, some models do actually characterize the epistemic status of attributed statements. 
For instance, nanopublications package an assertion together with provenance and publication information \cite{GrothGibsonVelterop2010Nanopublication,KuhnDumontier2017Nanopublications}. Micropublications extend this direction by representing scientific claims together with evidence, support, challenge, and argument structure \cite{ClarkCiccareseGoble2014Micropublications}. CIDOC CRM and its extensions, especially CRMinf, provide another model for representing argumentation, inference making, belief adoption, and provenance in cultural heritage contexts \cite{Doerr2003CIDOCCRM,CRMinf}. 
Large knowledge graphs often already contain statements whose factuality is mediated by references, qualifiers, ranks, source claims, or disagreement, as in Wikidata \cite{vrandecic2014, WikidataRanking}. However, Wikidata's modeling conventions do not amount to a general theory of epistemic statuses: referenced statements may be well sourced, deprecated, disputed, or merely reported, but the semantics of these distinctions is only partially formalized \cite{dipasquale2024weakerLogicalStatusWikidata}.
Some works have begun to address this gap more directly. For instance, eSPARQL extends SPARQL-star with epistemic querying over RDF-star metadata, using a four-valued logic \cite{PanHernandezSeiferLaemmelStaab2024ESPARQL}. Here, explicitly, an embedded statement is not automatically expected to be factual: it can be believed true, believed false, believed unknown, or believed conflicting by an individual or group. eSPARQL also supports nested beliefs and queries over conflicts between belief sets. By packaging assertions with evidence and provenance, these models are able to express complex networks of relations and epistemic statutes for basic triples, but a general characterization of the epistemic nature of these assertions is still missing.

DEC addresses this missing layer. It does not replace RDF-star, named graphs, nanopublications, CRMinf, or contextual reasoning frameworks. Instead, it provides a semantic characterization of the cognitive containers in which attributed statements occur. Its main interest is the relation between the presence of a statement in a graph and its factual commitment, characterizing not only where a statement comes from, but also what kind of epistemic world the statement belongs to and how, if at all, it may interact with the factual core.

\section{Conclusions}\label{sec:conclusions}

This paper addressed the semantic gap between the widespread representation of provenance-enhanced statements in knowledge graphs and the lack of a principled account of how such statements relate to factual commitment and to one another. We proposed DEC as a semantic framework that interprets provenance relations as indicators of epistemic stance and organizes provenance-homogeneous bundles of claims into cognitive worlds. By grounding cognitive worlds in a family of cognitive modal logics (doxastic, epistemic, and conjectural) and by making locality and permeation explicit, DEC supports the coexistence of competing perspectives without collapsing them into inconsistency, while still enabling ordinary RDF/RDFS/OWL inference where factual commitment is warranted.

We formalized a DEC interpretation for RDF datasets that remains conservative with respect to standard RDF semantics, and we discussed how intensionality and locality provide a principled resolution of classical substitution problems (including identity-related paradoxes) that have historically limited the use of epistemic reasoning in Semantic Web settings. On this basis, we showed how epistemically meaningful phenomena such as disagreement, delusion, settlement, and personal opinions can be characterized in terms of (i) the world typing induced by governing predicates, (ii) controlled permeation to and from a distinguished reality graph, and (iii) optional conflict checks based either on the host reasoner (logical inconsistency) or on explicit epistemic constraints (closure).

To demonstrate feasibility, we described a prototype DEC reasoner implemented as a supervisory layer on top of Apache Jena Fuseki. The prototype recognises cognitive worlds across different RDF representation styles (named graphs, quoted triples, and reification), applies the permeation policies associated with epistemic types, delegates standard entailment to an underlying reasoner, and can annotate detected conflicts so that they become queryable parts of the dataset. This shows that DEC can be adopted incrementally, without redefining RDF/OWL entailment, and can be integrated into existing infrastructures as a dataset-level service.

Several directions remain open. First, scalability and deployment trade-offs deserve systematic evaluation: conflict detection across many worlds and aggressive materialisation strategies may be costly, and future implementations should explore more selective caching and lazy computation of permeated consequences. Second, richer forms of epistemic constraint beyond predicate-level closure (e.g., constraints scoped to classes, shapes, or contextual patterns) could improve disagreement detection in domains where incompatibility is not reducible to single properties. Third, the present work focused on permeation between each world and reality; extending DEC with principled inter-world permeation policies, temporal evolution of worlds, and explanation facilities for epistemic annotations would further support real-world scholarly and scientific workflows.

\section*{Acknowledgments}\label{sec:ack}
TBD

\bibliography{Content}

\appendix
\section{Appendix: Kripke semantics and cognitive modal logics}\label{app:kripke}

Section~\ref{sec:modallogics} introduces cognitive modal logics axiomatically, in order to keep the focus on the inferential roles of the axioms. This appendix provides the corresponding semantic layer: it is based on the usual Kripke-style evaluation of $\Box_i$, expresses the axioms of normal modal logic in terms of the frame conditions they induce, and justifies axiom \textbf{C} under a non-bivalent reading of the base logic. 

Let $\mathcal{L}_0$ be a chosen base logic (we will implicitly use propositional logic, but the same construction applies to first-order logics and description logics). Let $W$ be a non-empty set of worlds and, for each agent or cognitive stance $i$, let $R_i \subseteq W \times W$ be an accessibility relation.
A \emph{(multi-agent) Kripke frame} is a structure $\langle W,\{R_i\}_{i\in I}\rangle$. 
A \emph{Kripke model} is $\mathcal{M}=\langle W,\{R_i\}_{i\in I},V\rangle$, where $V$ is a valuation for the non-modal language of $\mathcal{L}_0$ at each world.

Since the conjectural setting presupposes that the base logic may leave some formulas \emph{unentailed} (rather than forcing bivalence), we allow the valuation to be non-bivalent. Concretely, we assume that for each world $w\in W$ and each base formula $\varphi\in\mathcal{L}_0$, $V_w(\varphi)$ is a partial function that takes values in $\{\top,\bot\}$. Formulas that do not belong to the inverse image of $V_w$ are said to be \textit{undefined/unentailed} at $w$.

Satisfaction is defined in the usual way for truth:
\[
\mathcal{M},w \models \varphi \quad\text{iff}\quad V_w(\varphi)=\top.
\]

Modal satisfaction uses the standard Kripke clause:
\[
\mathcal{M},w \models \Box_i \varphi
\quad\text{iff}\quad
\text{for all }w' \text{ such that } wR_i w',\ \mathcal{M},w' \models \varphi.
\]
The dual operator $\Diamond_i$ is definable as usual, but it plays no essential role in our cognitive setting. The axioms introduced in normal modal logics correspond to well-known structural constraints on the accessibility relations $R_i$ of the Kripke frame. In particular:

\begin{itemize}

\item \textbf{Axiom K.}  
Axiom \textbf{K} is valid on all Kripke frames and guarantees that implication is preserved under $\Box_i$ within the set of $R_i$-accessible worlds.

\item \textbf{Axiom D.}  
$\forall w \in W \ \exists w' \in W, w R_i w'$. \\
Axiom \textbf{D} corresponds to seriality of $R_i$ and excludes vacuous necessity, enforcing modal consistency.

\item \textbf{Axiom T.}  
$ \forall w \in W, w R_i w$. \\
Axiom \textbf{T} corresponds to reflexivity of $R_i$ and ensures that whatever holds under the modality also holds at the base level of the same world.

\item \textbf{Axiom 4.}  
$\forall w,u,v \in W, 
(w R_i u \wedge u R_i v) \Rightarrow w R_i v$. \\
Axiom \textbf{4} corresponds to transitivity of $R_i$ and guarantees positive introspection: if a formula is accepted under the modality, then its acceptance is itself available under further applications of the same modality.

\item \textbf{Axiom 5.}  
$\forall w,u,v \in W, 
(w R_i u \wedge w R_i v) \Rightarrow u R_i v$. \\
Axiom \textbf{5} corresponds to the Euclidean property of $R_i$ and guarantees negative introspection: whenever a formula fails to be accepted under the modality, this failure is itself accessible and therefore available under further applications of the same modality.
\end{itemize}

In addition, \cite{Vit26} introduces a new axiom, called Axiom \textbf{C}, which is not a frame condition in the classical sense, but a semantic constraint on how accessibility interacts with truth evaluation, and it specifically relies on non-bivalent base logics to avoid modal collapse.

\begin{itemize}
\item \textbf{Axiom C.}  
$\forall w,u \in W, w R_i u \Rightarrow \forall \varphi,  (\mathcal{M}, w \models \varphi \Rightarrow \mathcal{M}, u \models \varphi) $. \\
\textbf{C} is a \emph{truth-preservation condition} along $R_i$: any formula that is true at a world must remain true in all $R_i$-accessible worlds. 
\end{itemize}

Soundness and completeness of the system \textbf{KC} are established within a semantic framework specifically designed for conjectural reasoning over non-bivalent base logics. Rather than relying on standard frame-based correspondence alone, the analysis makes essential use of truth-preservation along accessibility relations,  capturing the intended cognitive reading of axiom \textbf{C}.

\bigskip
A logical system is sound if every formula derivable in the system is valid in all models of its intended semantics, and it is complete if every formula that is valid in all models of its intended semantics is derivable in the system.

\paragraph{Soundness.}
Soundness of \textbf{KC} follows directly from the semantic interpretation of the modal operator. Axiom \textbf{K} is sound under the standard Kripke clause for $\Box_i$, independently of any constraint on $R_i$. Axiom \textbf{C} is sound provided that accessibility preserves established truths of the base logic: whenever a base formula is true at a world, it remains true at all $R_i$-accessible worlds. Under this reading, every inference licensed by \textbf{KC} preserves truth in all cognitive models satisfying the corresponding persistence condition.

\paragraph{Completeness.}
Completeness of \textbf{KC} is proved by constructing a canonical cognitive model tailored to the conjectural setting. Worlds in the canonical model represent coherent cognitive states of the base logic, allowing for formulas to be unentailed rather than forced to be true or false. The accessibility relation is defined so as to preserve all established base truths, thereby enforcing axiom \textbf{C} at the semantic level.

\bigskip

Thus, any modal formula not derivable in \textbf{KC} fails at some canonical cognitive world, while every derivable formula is valid in all such worlds. This establishes that \textbf{KC} is both sound and complete with respect to the class of cognitive models in which accessibility is interpreted as truth-preserving informational extension, rather than mere frame-based reachability.

Adding axioms \textbf{D}, \textbf{4}, and \textbf{5} preserves soundness and completeness, as they impose only standard relational constraints on $R_i$ and do not interfere with the truth-preservation condition required by axiom \textbf{C}. Consequently, systems \textbf{KDC} and \textbf{KDC45} remain sound and complete. Detailed proofs can be found in \cite{Vit26}. 

\paragraph{Accessibility as information growth in non-bivalent cognitive semantics}

In cognitive settings we interpret accessibility as growth of information. Intuitively, $w R_i w'$ means that, from the perspective of agent $i$, $w'$ is obtained from $w$ by adding determinations for formulas that were previously undefined, while preserving what was already established in $w$.

In the following, $\mathrm{defined}$ is the set of all defined formulas of a world, or
\[
\mathrm{defined}(w)=\{ \varphi \mid V_w(\varphi)\in \{\top,\bot\} \ \} 
\]

We make this concept precise, by expressing two coherence principles that do not presuppose bivalence.

\begin{itemize}
\item \emph{Monotonicity of definedness} (MD): if $w R_i w'$ then $\mathrm{defined}(w)\subseteq \mathrm{defined}(w')$. No formula that was defined at $w$ becomes undefined at $w'$.
\item \emph{Persistence of established truths} (PT): if $w R_i w'$ and $V_w(\varphi)=\top$, then $V_{w'}(\varphi)=\top$. Truths cannot be lost along $R_i$.
\end{itemize}

MD captures the idea that accessibility adds information rather than removing it. PT captures the idea that accessibility is conservative over what is already accepted as true. Together they characterize accessibility as an \emph{extension} of the informational state. In particular, MD permits $V_w(\varphi) \notin \{\top,\bot\}$ and $V_{w'}(\varphi)\in \{\top,\bot\}$, which is the formal way to say that conjectures or observations may settle previously undefined formulas at accessible worlds.

These semantic constraints explain the role of the modal axioms without appealing to bivalence. PT makes Axiom C sound, since $V_w(\varphi)=\top$ entails $V_{w'}(\varphi)=\top$ for every $w'$ with $w,R_i,w'$, and therefore $\varphi \to \Box_i\varphi$ is valid under the intended reading. MD and PT together imply that $R_i$ connects $w$ only to worlds that agree with $w$ on everything already defined, and possibly define more. If one also requires persistence of falsity, then $w \, R_i \, w'$ means agreement on the entire defined part of $w$, which is the strongest form of informational extension.

This reading aligns with the cognitive families we use. In doxastic worlds, $R_i$ collects the situations the agent takes as admissible extensions of the current belief state, so beliefs are evaluated for stability under growth of information. In epistemic worlds, $R_i$ collects the situations that remain indistinguishable for the agent given what is already known, and PT guarantees that what is known does not get lost when additional details are fixed. In both cases, accessibility is not a jump to an arbitrary alternative world, but it is a move to a world that preserves and refines the current informational core.

\end{document}